\documentclass[final]{siamonline190516}
\usepackage{lipsum}
\usepackage{amsfonts}
\usepackage{graphicx}
\usepackage{epstopdf}
\usepackage{algorithmic}
\ifpdf
  \DeclareGraphicsExtensions{.eps,.pdf,.png,.jpg}
\else
  \DeclareGraphicsExtensions{.eps}
\fi

\usepackage{enumitem}
\setlist[enumerate]{leftmargin=.5in}
\setlist[itemize]{leftmargin=.5in}

\newsiamremark{remark}{Remark}
\newsiamremark{hypothesis}{Hypothesis}
\crefname{hypothesis}{Hypothesis}{Hypotheses}
\newsiamthm{claim}{Claim}

\headers{On stochastic discrete-time SIS and SIR models}{Schreiber, Huang, Jiang, and Hao}

\title{Extinction and quasi-stationarity for \\discrete-time, endemic SIS and SIR models}

\author{
 Sebastian Schreiber\thanks{Department of Evolution and Ecology and Center for Population Biology, University of California, Davis, California 95616, United States (\email{sschreiber@ucdavis.edu})}
 \and
  Shuo Huang\thanks{Mathematics and Science College, Shanghai Normal University, Shanghai 200234, China (\email{1000441256@smail.shnu.edu.cn}, \email{jiangjf@shnu.edu.cn})}
 \and
   Jifa Jiang$^\dagger$
 \and
  Hao Wang\thanks{Department of Mathematical and Statistical Sciences, University of Alberta, Edmonton, Alberta T6G 2G1, Canada (\email{hao8@ualberta.ca})}
  }

\usepackage{amsopn}

\newcommand{\eig}{\alpha}
\newcommand{\vs}{\vskip 0.05in}
\usepackage{natbib}

\begin{document}

\maketitle

\begin{abstract}
Stochastic discrete-time SIS and SIR models of endemic diseases are introduced and analyzed. For the deterministic, mean-field model, the basic reproductive number $R_0$ determines their global dynamics. If $R_0\le 1$, then the frequency of infected individuals asymptotically converges to zero. If $R_0>1$, then the infectious class uniformly persists for all time; conditions for a globally stable, endemic equilibrium are given. In contrast, the infection goes extinct in finite time with probability one in the stochastic models for all $R_0$ values. To understand the length of the transient prior to extinction as well as the behavior of the transients, the quasi-stationary distributions and the associated mean time to extinction are analyzed using large deviation methods. When $R_0>1$, these mean times to extinction are shown to increase exponentially with the population size $N$. Moreover, as $N$ approaches $\infty$, the quasi-stationary distributions are supported by a compact set bounded away from extinction; sufficient conditions for convergence to a Dirac measure at the endemic equilibrium of the deterministic model are also given. In contrast, when $R_0<1$, the mean times to extinction are bounded above $1/(1-\alpha)$ where $\alpha<1$ is the geometric rate of decrease of the infection when rare; as $N$ approaches $\infty$, the quasi-stationary distributions converge to a Dirac measure at the disease-free equilibrium for the deterministic model. For several special cases, explicit formulas for approximating the quasi-stationary distribution and the associated mean extinction are given. These formulas illustrate how for arbitrarily small $R_0$ values, the mean time to extinction can be arbitrarily large, and how for arbitrarily large $R_0$ values, the mean time to extinction can be arbitrarily large. 
\end{abstract}

% REQUIRED
\begin{keywords}
infectious diseases, discrete-time SIS model, discrete-time SIR model, times to extinction, quasi-stationary distributions, large deviations
\end{keywords}

% REQUIRED
%\begin{AMS}
%test
%\end{AMS}

\section{Introduction}

Infectious disease modeling has been one of the most important topics in mathematical biology~\citep{Keeling2011}. A recent Google scholar search\footnote{On September 8th, 2019, a search on Google scholar with the search term ``disease model SIR OR SIS'' returned $1.45$ million results.} reveals over a million studies referencing SIS (susceptible-infected-susceptible) and SIR (susceptible-infected-recovered) models. Most of these studies use deterministic, continuous-time equations. However, in seasonal systems or systems where measurements are taking at regular time intervals (e.g. day or week), discrete-time models play an important role~\citep{anderson_may1991,allen1994,allen_burgin2000,klepac_pomeroy2009,Keeling2011}. For many of these models, the basic reproductive number, $R_0$, determines whether a disease can persist or not in a population. If $R_0>1$, persistence often occurs. While if $R_0<1$, the disease-free state (i.e. extinction) often is globally stable and the infection is lost asymptotically as time marches to infinity.

When considering finite populations without external sources of infection, Markov chain models typically predict that the disease goes extinct in finite-time whether $R_0>1$ or $<1$. To understand this puzzling difference between the asymptotic behaviors of the deterministic and stochastic models~\citep{bartlett1966introduction,Keeling2011,diekmann2012mathematical}, one can use the concept of quasi-stationarity that describes the long-term behavior of the stochastic model conditioned on non-extinction~\citep{darroch_seneta1965,Darroch1967}. For finite state models, the quasi-stationary distribution corresponds to a normalized left eigenvector $\pi$ of the transition matrix of the Markov chain restricted to the transient states, i.e. the states where the infection persists. In discrete-time, if the disease dynamics follow the quasi-stationary distribution, then the eigenvalue $\lambda$ associated with this eigenvector corresponds to the probability of disease persistence over the next time step (respectively, a time interval of length one). Thus, when the stochastic model  follows the quasi-stationary distribution, the time to extinction is exponentially distributed with a mean time of extinction $1/(1-\lambda)$. \citet{GW} call $1/(1-\lambda)$ the intrinsic mean time to extinction. To understand the link between the stochastic and deterministic models, it is natural to ask: How does the intrinsic mean time to extinction increase as the population size gets larger? How is the quasi-stationary distribution related to the asymptotic dynamics of the deterministic model as the population size gets larger? More generally, how do these quantities depend on the parameters such as $R_0$? 

For continuous-time, stochastic SIS models,  there exists a dichotomy in the mean time to extinction when a fixed, positive fraction of the population is infected~\citep{weiss_dishon1971,barbour1975,kryscio_lefevre1989}. When $R_0>1$, this time increases exponentially with the population size $N$ in the limit of large population sizes. When $R_0<1$, these extinction times remain bounded in the limit of large population size. However, to the best of our knowledge, similar statements for the intrinsic mean extinction times have not been rigorously proven for these continuous-time SIR models. However, in a series of papers~\citep{naasell1996,naasell1999,naasell2001,naasell2002}, N{\aa}sell provided methods to approximate the intrinsic mean extinction times as well as the quasi-stationary distributions. His approximations support the existence of a similar dichotomy for intrinsic mean times to extinction. Moreover, they highlight a remarkable dichotomy about the qualitative behavior of the quasi-stationary distributions. When $R_0>1$, these distributions are well-approximated by a normal distribution centered near the endemic equilibrium. When $R_0<1$, the quasi-stationary distribution is best approximated by a discrete, geometric distribution. Despite these advances, mathematically rigorous results for discrete-time, stochastic SIS and SIR models are lacking.

In this paper, we introduce a new class of discrete-time SIS and SIR deterministic and stochastic models that have several desirable properties including (i) they are derived with individual-based rules and, consequently, preserve non-negativity of all populations, (ii) the deterministic models are the mean field model of the stochastic models, and (iii) the deterministic models converge to the classical continuous-time models in an appropriate limit. For these models, we  analyze the global dynamics of deterministic models and then use this analysis in conjunction with large deviation results from \citet{FS} to rigorously characterize the behavior of the intrinsic mean times to extinction and quasi-stationary distributions in the limit of large population sizes. Moreover, for some special cases, we derive explicit approximations for the quasi-stationary distributions and extinction times that apply for all population sizes.

Our paper is structured as follows. Section 2 introduces and analyzes the discrete-time, deterministic SIS  model. Section 3 presents mathematical and numerical findings on quasi-stationary distributions and intrinsic mean extinction times for the stochastic SIS model. Section 4 introduces the discrete-time SIR model and proves results about its global attractors. Section 5 presents mathematical and numerical findings on quasi-stationary distributions and intrinsic mean extinction times for the stochastic SIR model. Section 6 discusses our main findings and future challenges. The mathematical proofs are given in Sections 7 through 10. 

\section{The dynamics of a deterministic SIS model}
We begin with a discrete-time version of the classical susceptible-infected-susceptible (SIS) model where individuals are either susceptible or infected. Let $I_n$ denote the fraction of individuals that are infected at time step $n$, in which case, the fraction of susceptible individuals equals $1-I_n$. Individuals escape natural mortality with probability $e^{-\mu}$ while infected individuals escape recovery with probability $e^{-\gamma}$, where $\mu>0$ and $\gamma>0$. Susceptible individuals from the previous time step who have not died, escape infection with probability $e^{-\beta I_n}$ where $\beta>0$ is the contact and transmission rate. If the population size remains constant, then the disease dynamics are given by
\begin{equation}\label{DSIS}
    I_{n+1}=F(I_n):=e^{-\mu}(1-I_n)\big(1-e^{-\beta I_n}\big)+e^{-\mu-\gamma}I_n,\ 0\le I_n \le 1.
\end{equation}
This discrete-time formulation of the SIS model has several advantages. First, it is straightforward to verify that the dynamics of $I_n$ remain in the interval $[0,1]$ provided the initial value $I_0$ lies in this interval. Second, these dynamics, as described in the next section, correspond to the mean field dynamics of an individual-based model. Finally, if $\Delta t$ is the length of a time step, and $\beta=\tilde \beta \Delta t$, $\gamma=\tilde \gamma \Delta t$, $\mu=\tilde \mu \Delta t$, then
\[
I(t+\Delta t):= I_{n+1}= (1-I(t)) \tilde \beta \Delta t + I(t) -(\tilde \mu +\tilde \gamma )\Delta t I(t) +O(\Delta t^2) \mbox{ where }I(t):=I_n
\]
Hence, in the limit $\Delta t\to 0$, we get the classical SIS ordinary differential equation
\[
\frac{dI}{dt}=\lim_{\Delta t\to 0}\frac{I(t+\Delta t)-I(t)}{\Delta t}=(1-I)\tilde \beta I -(\tilde \mu +\tilde \gamma)I.
\]

To understand the dynamics of \eqref{DSIS}, we can linearize at the origin to obtain the per-capita growth rate of the infection at the disease free-equilibrium
\begin{equation}\label{BRN}
    \eig= \eig(\mu,\beta,\gamma):=\beta e^{-\mu}+ e^{-\mu-\gamma}.
\end{equation}
The basic reproduction, alternatively, is given by
\begin{equation}\label{BRN2}
R_0=\beta e^{-\mu}/(1-e^{-\mu-\gamma}).
\end{equation}
As $\eig>1$ if and only if $R_0>1$, we can use the basic reproductive number to characterize the global dynamics, as the following theorem shows. \vs

\begin{theorem}\label{SISDICH}
{\rm (i)}If $R_0\le 1$, then the origin is globally asymptotically stable.

{\rm (ii)}If $R_0>1$, then there is a unique positive fixed point in $(0,1]$ such that it is  globally asymptotically stable in $(0,1]$.

\end{theorem}

\newcommand{\s}{\mathcal{S}}

\section{Metastability and extinction in a stochastic SIS model}
\begin{figure}
\includegraphics[width=0.45\textwidth]{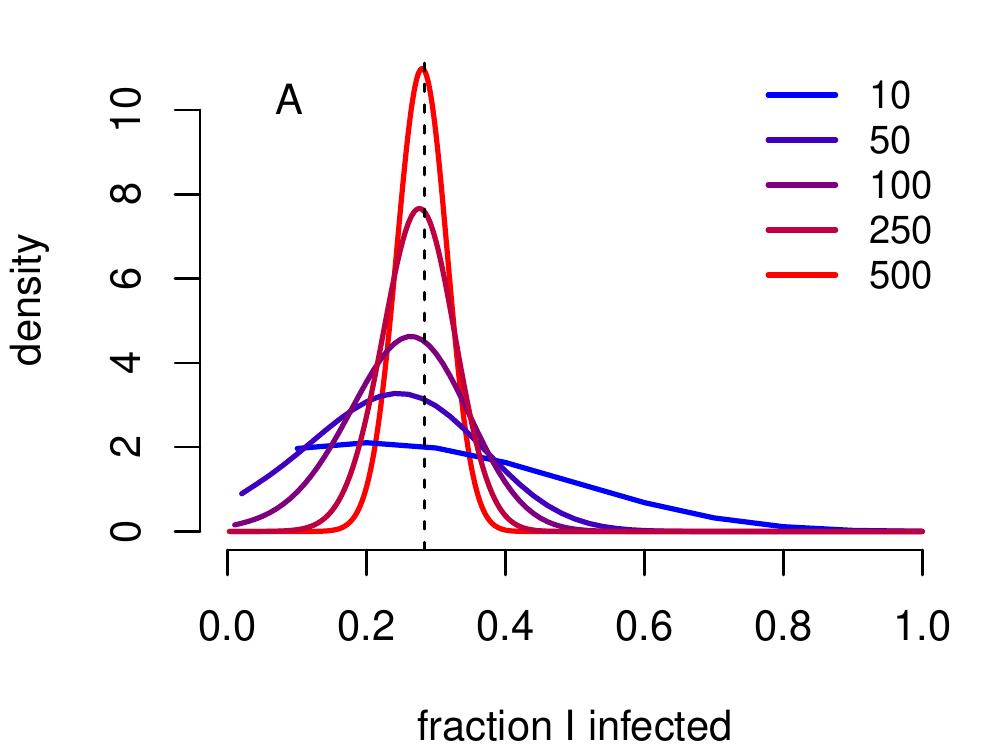}
\includegraphics[width=0.45\textwidth]{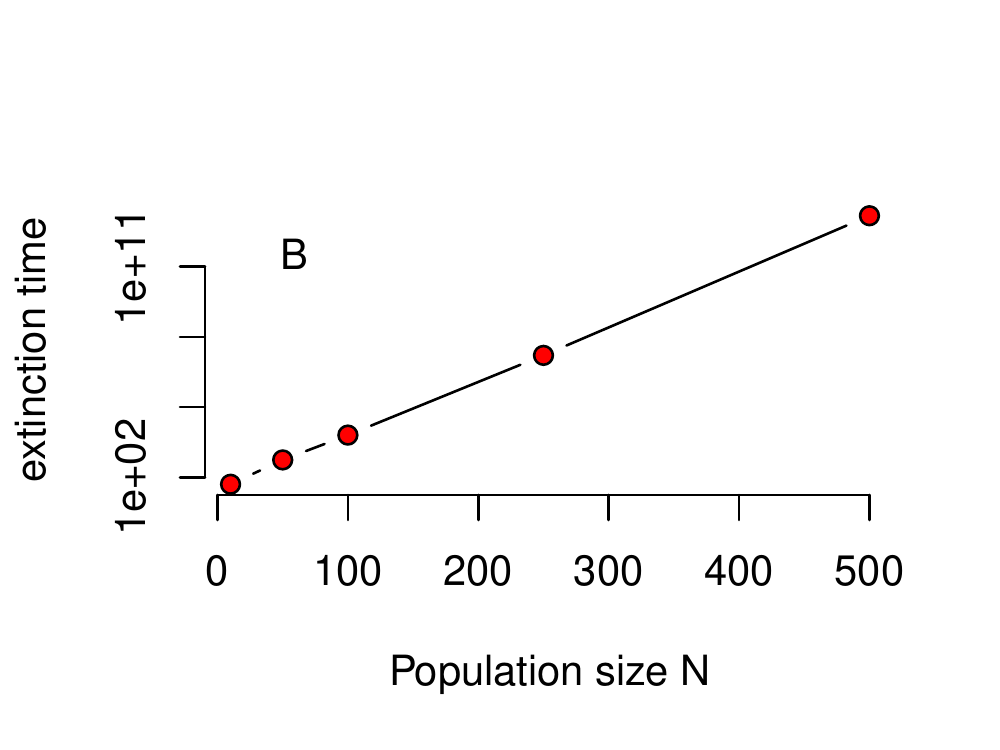}\\
\includegraphics[width=0.45\textwidth]{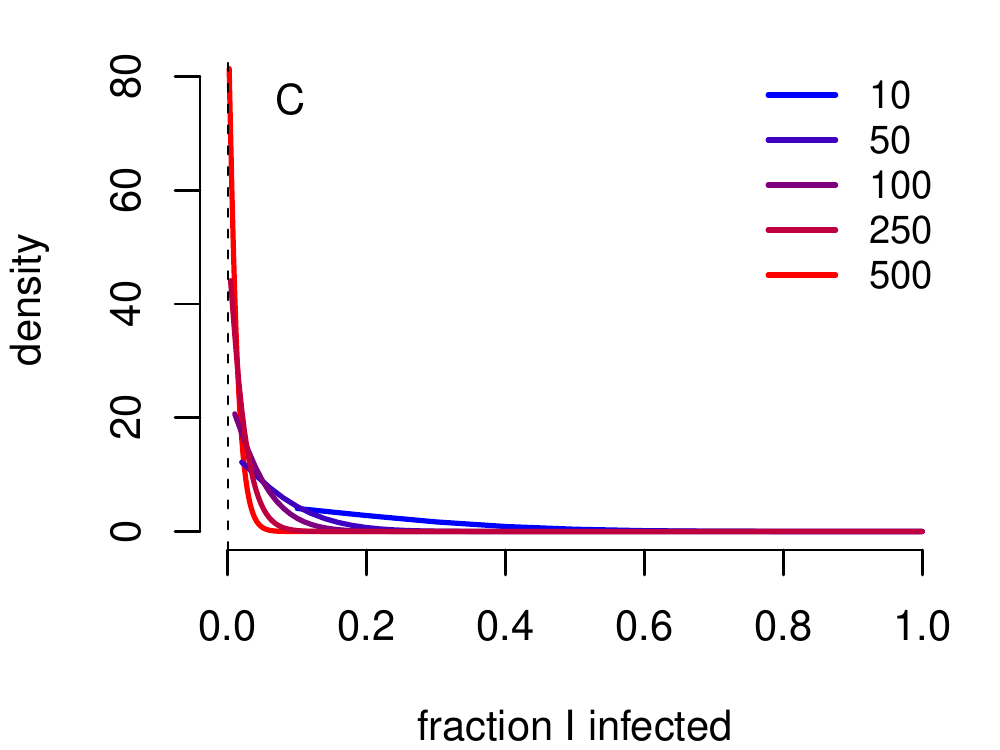}
\includegraphics[width=0.45\textwidth]{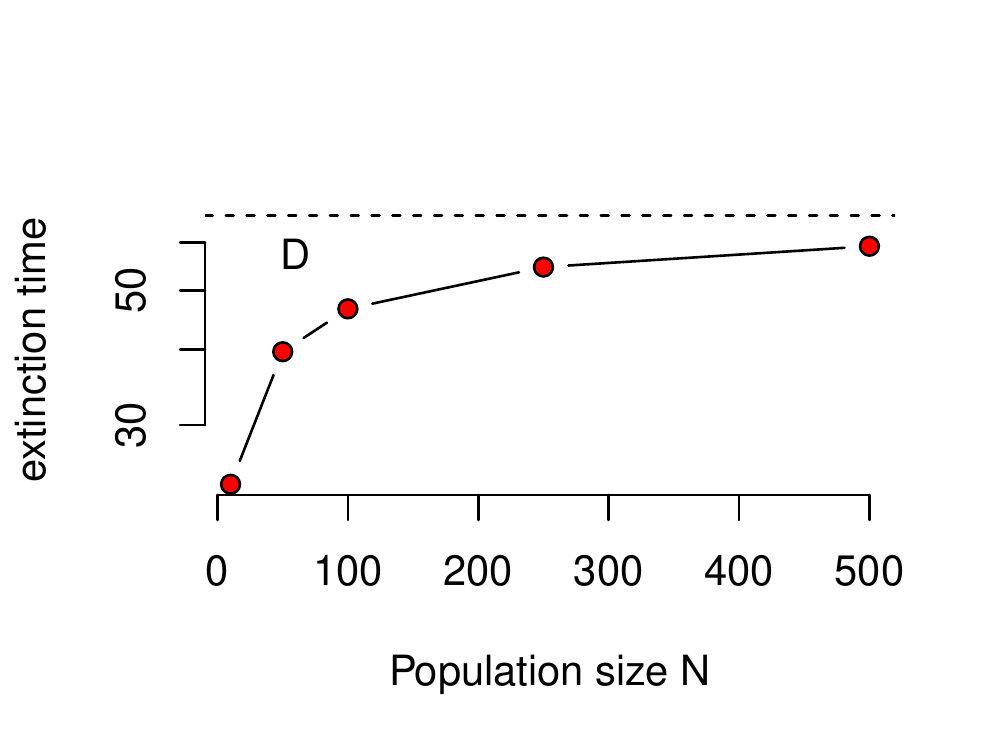} 
\caption{Quasi-stationary distributions and intrinsic mean extinction times for the stochastic SIS model for $R_0>1$ (A,B) and $R_0<1$ (C,D). In (A) and (C), the quasi-stationary distributions approach a Dirac distribution at the equilibrium density (dashed line). In (B), the intrinsic mean time to extinction increases exponentially with population size for $R_0>1$. In (D), the intrinsic mean time to extinction saturates at $1/(1-\eig)$ (dashed line) as population size increases. Parameter values: $\gamma=0.1$, $\mu=0.01$, and $\beta=0.15$ for (A,B) and $\beta=0.09$ for (C,D).}\label{fig:SSIS}
\end{figure}

For the individual-based stochastic model, we require the additional parameter of the total population size $N$. Given this population size, the state space corresponds to the possible fractions of infected individuals in the population
\[
\s=\left\{0,\frac{1}{N},\frac{2}{N},\dots,\frac{N-1}{N},1\right\}.
\]
Let $I_n\in \s$ be the fraction of individuals infected at time step $n$. To determine the fraction infected in the next time step, we assume that each infected individual remains infected with probability $e^{-\mu-\gamma}$ independent of each other and each susceptible individual lives and becomes infected with probability $e^{-\mu}(1-e^{-\beta I_n})$ independent of each other. Hence,
\begin{equation}\label{SSIS}
\begin{aligned}
I_{n+1}&=\frac{X_{n+1}+Y_{n+1}}{N} \mbox{ where}\\
X_{n+1}&\sim \rm{Binom}(NI_n,e^{-\mu-\gamma}) \mbox{ and}\\
Y_{n+1}&\sim\rm{Binom}\Big(N(1-I_n),e^{-\mu}\big(1-e^{-\beta I_n}\big)\Big).
\end{aligned}
\end{equation}
\renewcommand{\P}{\mathbb{P}}
\newcommand{\E}{\mathbb{E}}
For all $1\le i,j \le N$, let $Q_{ij}=\P[X_{n+1}=j/N|X_n=i/N]$ be the transition probabilities restricted to the transient states $\s\setminus\{0\}$ and $Q=(Q_{ij})$ be the associated $N\times N$ matrix. Unlike the deterministic model, the disease goes extinct in finite time with probability one for the stochastic model. The following proposition follows from standard results in stochastic processes~\citep[see e.g.][Theorem 3.20]{harier-18}.

\begin{proposition}\label{prop1} Assume that $\mu+\gamma$ and $\beta$ are positive.  With probability one, $I_n=0$ for some $n\ge 1.$
\end{proposition}

Even though extinction is inevitable, it may be preceded by long term transients. To characterize these transients, we use quasi-stationary distributions introduced by \citet{darroch_seneta1965}. As $Q$ is a sub-stochastic, positive matrix, there exists a dominant eigenvalue $\lambda\in (0,1)$ and associated dominant eigenvector  $\pi=(\pi_1,\dots,\pi_N)$ (depending on $N$) such that $\sum_i \pi_i=1$, $\pi_i>0$ for all $i$, and $\pi Q= \lambda \pi.$ $\pi$ is the \emph{quasi-stationary distribution} which satisfies \citep{darroch_seneta1965}
\[
\lim_{n\to\infty}\P[I_n=j/N|I_n>0]=\pi_j
\]
i.e. the probability of having $j$ individuals in the long-term given the disease hasn't gone extinct equals $\pi_j$. Furthermore,
\[
\sum_{i=1}^N \P[I_{n+1}>0|I_n=i/N]\pi_i=\lambda
\]
i.e. $\lambda$ is the probability the disease persists to the next time step given the process is following the quasi-stationary distribution. Hence, the mean time to extinction, when following the quasi-stationary distribution, is $\frac{1}{1-\lambda}$, what \citet{GW} call the \emph{intrinsic mean time to extinction}.

Our main result for the stochastic SIS model concerns the behavior of the quasi-stationary distribution and the intrinsic mean time to extinction for large population size $N$.\vs

\begin{theorem}\label{SISSICH} Assume $\mu+\gamma>0,\beta>0$. For each $N\ge 1$, let $\pi^N$ be the quasi-stationary distribution and $\lambda^N$ the corresponding eigenvalue. Let $\eig=\beta e^{-\mu}+e^{-\mu-\gamma}.$ Then
\begin{enumerate}
  \item[(i)]If $\eig\le 1$ (equivalently $R_0\le 1 $), then $\lambda^N\le \eig$ for all $N\ge 1$ and $\lim_{N\to\infty}\pi^N=\delta_0$ where $\delta_0$ is a Dirac measure at $0$ and convergence is in the weak* topology on probability measures on $[0,1]$. 
  \item[(ii)] If $\eig>1$ (equivalently $R_0>1$), then $\lim_{N\to\infty}\pi^N=\delta_{I^*}$ where $\delta_{I^*}$ is the Dirac measure at the unique positive fixed point $I^*$ of equation~\eqref{DSIS} and\\ $\limsup_{N\to\infty}\frac{1}{N}\log (1-\lambda_N)<0.$
\end{enumerate}
\end{theorem}\vs

The first assertion of Theorem~\ref{SISSICH} implies that if $\eig<1$ and the population size is large, then any long-term transients mostly involve low frequencies of infected individuals and the mean time to extinction after these transients is less than $\frac{1}{1-\eig}.$ We conjecture that in the limit of large population size, $N\to\infty$, the intrinsic mean time to extinction equals $\frac{1}{1-\eig}.$ The second assertion of Theorem~\ref{SISSICH} implies that if $\eig>1$ and the population size is large, then the long-term transients fluctuate around the equilibrium frequency of the deterministic model and the mean time to extinction following these transients increases exponentially with population size,  i.e. there exist $c_1,c_2>0$ such that $\frac{1}{1-\lambda^N}\ge c_1 e^{c_2N}$ for all $N\ge 1.$ Figure~\ref{fig:SSIS} illustrates these conclusions numerically.

\begin{figure}
\includegraphics[width=0.45\textwidth]{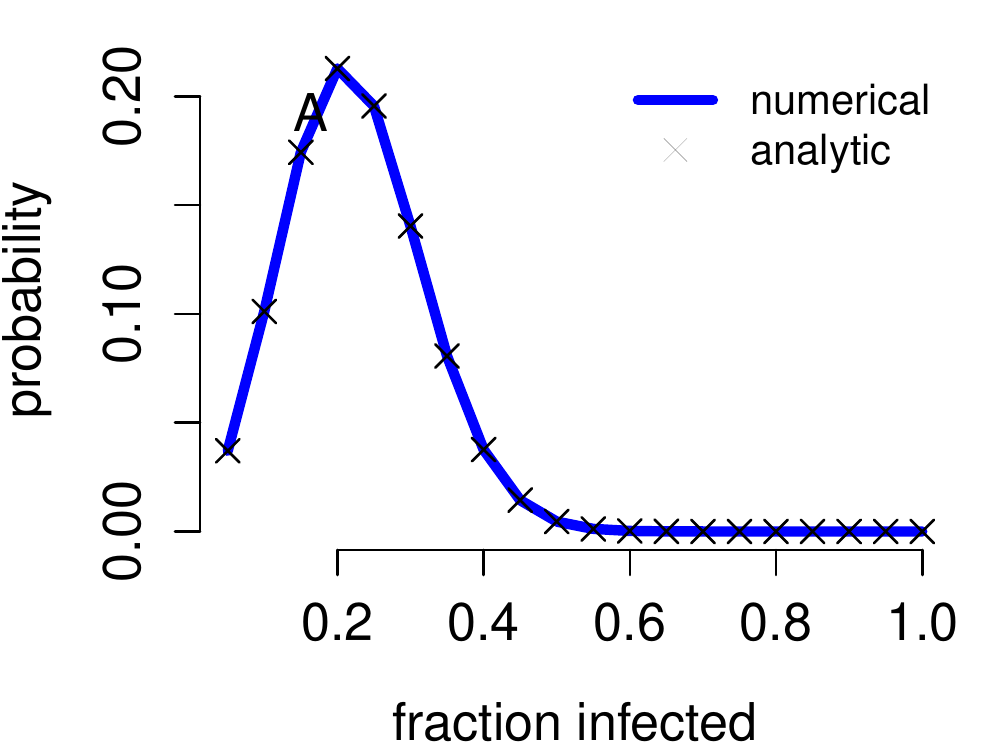}
\includegraphics[width=0.45\textwidth]{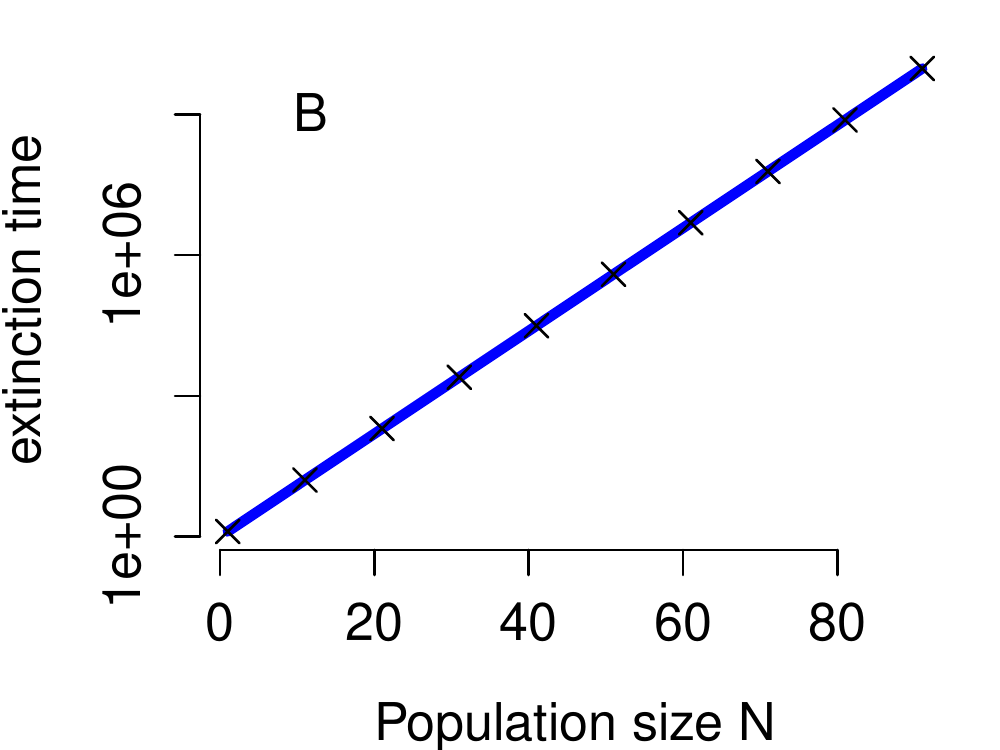}
\caption{Quasi-stationary distributions and intrinsic mean extinction times for the stochastic SIS model for high contact rates and low recovery rates. In (A), the numerically computed quasi-stationary distribution (x marks) and the analytical approximation (solid blue curve) $\pi_i ={N \choose i} e^{-\mu i}(1-e^{-\mu})^{N-i}/(1-e^{-\mu})^N$. In (B), the numerically computed intrinsic mean time to extinction  (x marks) and the analytical approximation (solid blue curve) $1/(1-(1-e^{-\mu})^N)$. Parameter values: $\gamma=0$, $\mu=1.5$, $\beta=100$, and $N=20$ for (A).}\label{fig:SSIS2}
\end{figure}

Given that Theorem~\ref{SISSICH} describes the effect on increasing population size on the intrinsic mean time to extinction for a fixed value of $\eig$, it is natural to ask what effect increasing $\eig$ has on these extinction times for a fixed population size. In general, this is a challenging question. However, we can answer this question for two special cases. First, we consider the case of low recovery rates $\gamma=0$ and very high $\beta\gg 1$ contact rates. In the limit of $\beta\to\infty$, the update rule for $I_n$ for $I_n>0$ is approximately $I_{n+1}\sim \rm{Binom}(N,e^{-\mu})$. Namely, provided there is at least one infected individual at time step $n$, all individuals that have not died get infected. In this case, the quasi-stationary distribution is approximately $\pi_i ={N \choose i} e^{-\mu i}(1-e^{-\mu})^{N-i}/\lambda^N$ for $i=1,2,\dots,N$ with the persistence eigenvalue $\lambda^N=(1-e^{-\mu})^N$. In particular, the mean intrinsic extinction time is bounded in the limit of $\eig\to \infty$ due to $\beta\to\infty$. The accuracy of this approximation is illustrated in Figure~\ref{fig:SSIS2}. Second, in the limit of no recovery and no mortality (i.e. $\mu=\gamma=0$), the disease (not surprisingly!) never goes extinct whenever $I_0>0.$ Indeed, in this case, $I_n\to 1$ as $n\to\infty$ with probability one provided $\beta>0$ and $I_0>0.$ These two special cases highlight that the effect of $\eig$ on the intrinsic mean time to extinction depends in a subtle way as $\eig$ increases.

\section{The dynamics of a deterministic SIR model}

As discrete-time counterpart to the classical susceptible-infected-recovered (SIR) model, we assume all individuals escape natural mortality with probability $e^{-\mu}$, infected individuals escape recovery with probability $e^{-\gamma}$, and susceptible individuals escape infection with probability $e^{-\beta I}$ where $I$ is the frequency of infected individuals and $\beta>0$ is the contact and transmission rate. If the population size is constant, then the discrete-time dynamics are given by
\begin{equation}\label{DSIR}
    \left\{\begin{array}{l}
                 S_{n+1}=1-e^{-\mu}+S_n e^{-\mu-\beta I_n}, \\
                 \noalign{\medskip}
                 I_{n+1}=e^{-\mu}S_n\big(1-e^{-\beta I_n}\big)+e^{-\mu-\gamma}I_n. \\
 \end{array}\right.
\end{equation}
Like the discrete-time SIS model, this discrete-time formulation of the SIR model has several advantages. First, by adding the two equations of (\ref{DSIR}) together, we obtain that the trajectories of (\ref{DSIR}) remain in the domain
$$X:=\{(S,I): S\ge 0,\ I\ge 0,\ S+I\le 1\}$$  provided the initial value $(S_0,I_0)$ lies in this domain. Furthermore, if we define $\partial{X_0}:=\{(S,0): 0\le S\le 1\}$ and $X_0:=X\setminus \partial{X_0}$, then $X_0$ and $\partial{X_0}$ are positively invariant. Second, these dynamics, as described in the next section, correspond to the mean field dynamics of an individual-based model. Finally, if $\Delta t$ is the length of a time step, and $\beta=\tilde \beta \Delta t$, $\gamma=\tilde \gamma \Delta t$, $\mu=\tilde \mu \Delta t$, then
\[
\begin{aligned}
S(t+\Delta t):=& S_{n+1}= \tilde \mu \Delta t + S(t) -(\tilde \mu +\tilde \beta I(t))\Delta t S(t)  +O(\Delta t^2) \mbox{ where }S(t):=S_n,\\
I(t+\Delta t):=& I_{n+1}= S(t)I(t)\tilde \beta \Delta t +I(t)-( \tilde \mu +\tilde \gamma)I(t)\Delta t+O(\Delta t^2) \mbox{ where }I(t):=I_n.\\
\end{aligned}
\]
Hence, in the limit $\Delta t\to 0$, we get the classical SIR system of ordinary differential equations
\[
\begin{aligned}
\frac{dS}{dt}&=\lim_{\Delta t\to 0}\frac{S(t+\Delta t)-S(t)}{\Delta t}=\tilde \mu -(\tilde \mu +\tilde \beta I)S\\
\frac{dI}{dt}&=\lim_{\Delta t\to 0}\frac{I(t+\Delta t)-I(t)}{\Delta t}=\tilde \beta  IS -(\tilde \mu +\tilde \gamma)I.
\end{aligned}
\]

For our discrete-time SIR model, the disease-free fixed point is $(1,0)$. At this fixed point, the per-capita growth rate of the disease still equals $\eig=\beta e^{-\mu}+e^{-\mu-\gamma}$ and the reproductive number still equals $R_0=\beta e^{-\mu}/(1-e^{-\mu-\gamma}.$ We will show that if $R_0>1$, the disease persists and if $R_0\le 1$, the disease-free equilibrium is globally stable. Furthermore, we will show when the recovery rate $\gamma$ is sufficiently small, there is a globally stable endemic equilibrium. To state these results precisely, we define the parameter space as $P:=\{\lambda:=(\mu,\beta,\gamma):\mu>0,\beta>0,\gamma\ge 0\}$. Let $C_P^0:=\{\lambda=(\mu,\beta,0) \in P:\eig(\lambda)>1\}$ be the parameters corresponding to no recovery ($\gamma=0$) and $\eig>1$. Finally, define
\[
C_P:=\{\lambda \in P:\eig(\lambda)>1,\ \mbox{\eqref{DSIR} admits a globally stable hyperbolic fixed point in}\, X_0 \}.
\]
Using these definitions, we have the following theorem.
\begin{theorem}\label{SIR}
{\rm (i)} If $\eig\le 1$, then the disease free fixed point $(1,0)$ is globally asymptotically stable.

{\rm (ii)} If $\eig>1$, then $F: X_0\rightarrow X_0$ admits a global and compact attractor contained in the interior of $X_0$.

{\rm (iii)} $C_P\supset C_P^0$ is a non-empty open subset in $P$.
\end{theorem}

In addition we conjecture that there is a globally stable endemic equilibrium when  $\mu=\tilde\mu \Delta t$, $\beta=\tilde \beta\Delta t$,  $\gamma=\tilde \gamma\Delta$, $\Delta t>0$ is sufficiently small, and $\eig>1$ (equivalently, $R_0>1$).

\section{Metastability and extinction in a stochastic SIR model}

\begin{figure}\centering
\includegraphics[width=0.45\textwidth]{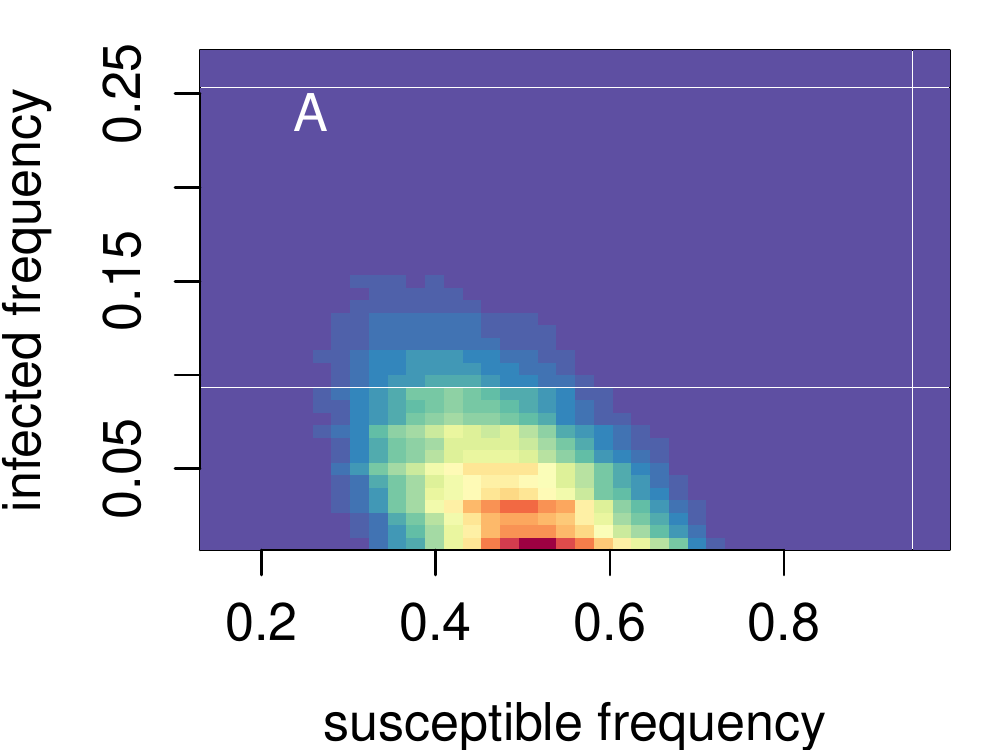}
\includegraphics[width=0.45\textwidth]{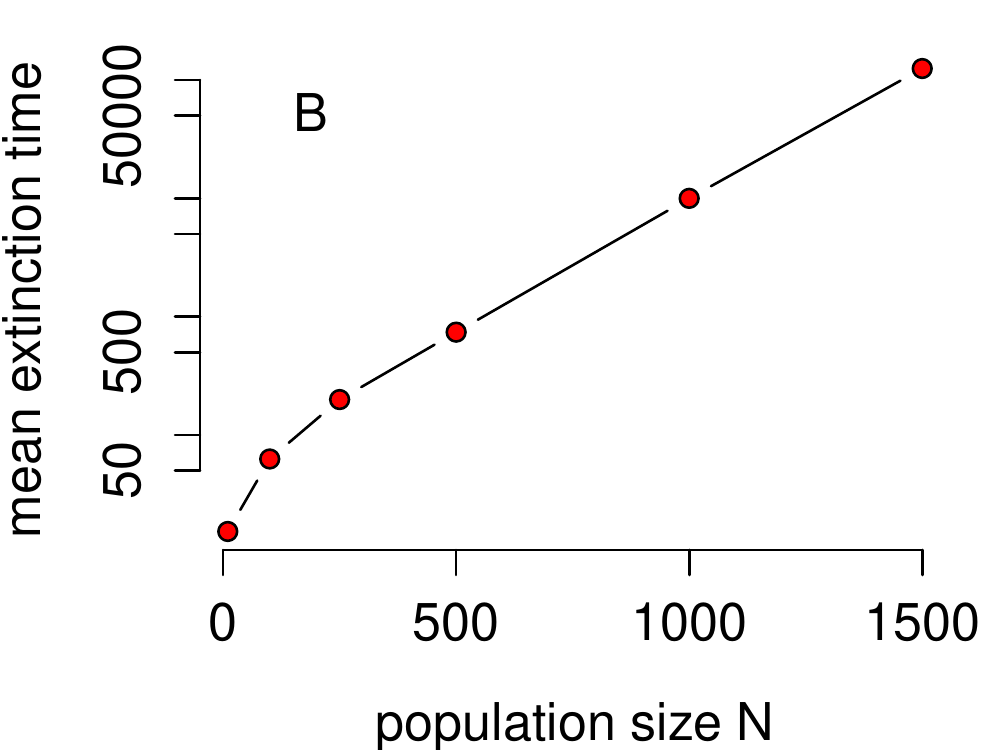}\\
\includegraphics[width=0.45\textwidth]{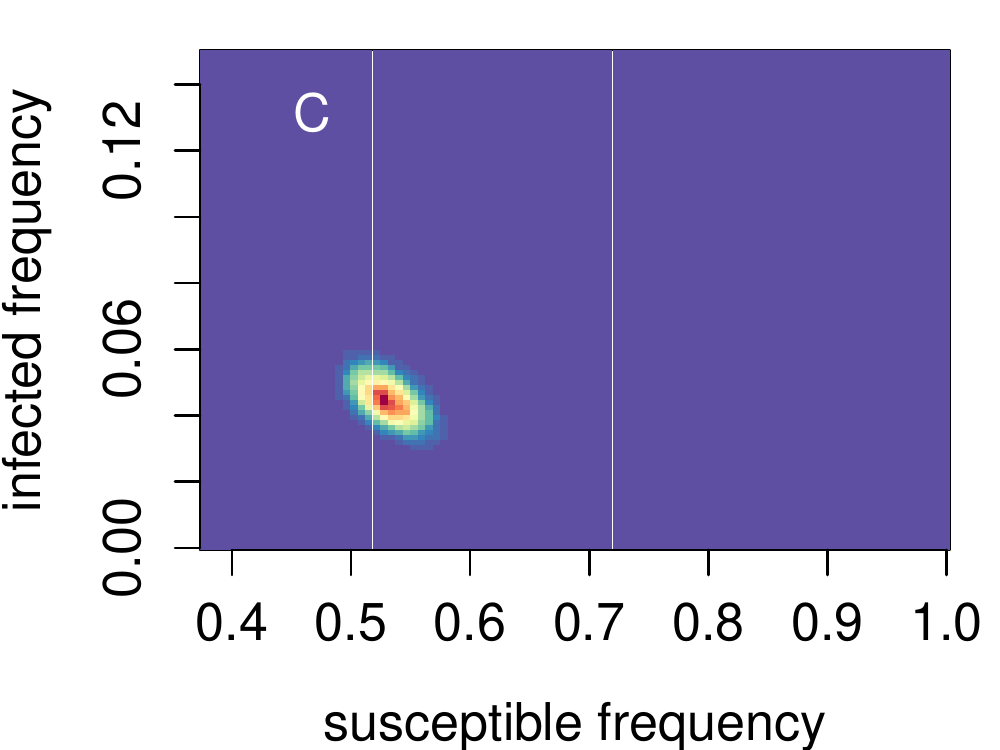}
\includegraphics[width=0.45\textwidth]{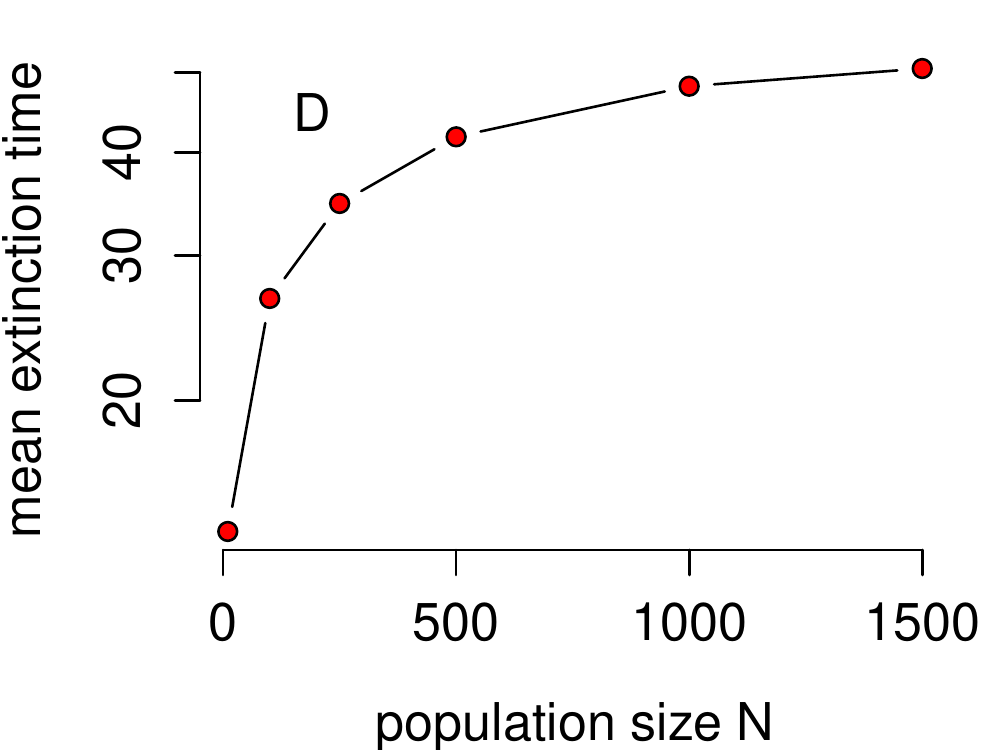}
\caption{Quasi-stationary distributions and mean intrinsic extinction times for the stochastic SIR model. For parameters with $\eig>1$, the quasi-stationary distributions, estimated numerically using the method of \citet{aldous-etal-88}, concentrate on the stable equilibrium as the population size goes from $N=100$ (A) to $N=10,000$ (C). For this $\eig>1$, the associated intrinsic mean extinction times increase exponentially with population size in (B). For parameters with $\eig<1$, the mean extinction times are bounded by $1/(1-\eig)$ in (D). Parameter values: $\mu=0.01$, $\gamma=0.1$ and $\beta=0.2$ for (A)-(C) and $\beta=0.09$ for (D). }\label{fig:SSIR}
\end{figure}

As with the stochastic SIS model, the stochastic SIR model requires the additional parameter of the total population size $N$. For a given population size, the state space $\s$ corresponds to the possible fraction of susceptible and infected individuals in the population,  i.e.
\[
\s=\{(i/N,j/N):i,j\in\{0,1,\dots,N\},i+j\le N\}\subset X.
\]
Let $(S_n,I_n)\in \s$ be the fractions of susceptible and infected individuals at time step $n$. The fraction of removed individuals at time step $n$ equals $R_n=1-S_n-I_n$. Consistent with the deterministic SIR model, we assume (i) each susceptible individual lives and becomes infected with probability $e^{-\mu}(1-e^{-\beta I_n})$ independent of each other, (ii) each infected individual either remains infected, dies and gets replaced with a susceptible individual, or enters the removed class with probabilities $e^{-\mu-\gamma}$, $1-e^{-\mu}$, or $e^{-\mu}(1-e^{-\gamma})$ independent of each other, and (iii) each removed individual dies and creates a new susceptible individual with probability $1-e^{-\mu}$. To account for these transitions, let $W_{n+1}$ be a binomial random variable with $NS_n$ trials and probability of success $e^{-\mu}(1-e^{-\beta I_n})$ (i.e.  susceptible individuals that will become infected), $X_{n+1}$ be a binomial random variable with $NI_n$ trails and probability of success $1-e^{-\mu}$ (i.e. infected individuals that die and get replaced by a susceptible individual), $Y_{n+1}$ be a binomial random variable with $NI_n-X_{n+1}$ trials with probability of success $e^{-\gamma}$ (i.e. non-dying infected individuals that will not enter the removed class), and $Z_{n+1}$ be a binomial random variable with $N(1-I_n-S_n)$ trials with probability of success $1-e^{-\mu}$ (i.e. removed individuals that die and get replaced with a susceptible individual). Under these assumptions, the stochastic SIR model is

\begin{equation}\label{SSIR}
\begin{aligned}
S_{n+1}&=\frac{1}{N}\left(NS_n-W_{n+1}+X_{n+1}+Z_{n+1}\right),\\
I_{n+1}&=\frac{1}{N}\left(W_{n+1}+Y_{n+1}\right) \mbox{, where}\\
W_{n+1}&\sim \rm{Binom}(NS_n,e^{-\mu}\big(1-e^{-\beta I_n}\big)),\\
X_{n+1}&\sim \rm{Binom}(NI_n,1-e^{-\mu}),\\
Y_{n+1}&\sim\rm{Binom}(NI_n-X_{n+1},e^{-\gamma}), \mbox{ and}\\
Z_{n+1}&\sim \rm{Binom}(N(1-I_n-S_n),1-e^{-\mu}).
\end{aligned}
\end{equation}

As with the stochastic SIS model, the disease goes extinct in finite time with probability one for the stochastic model. The following proposition follows from standard results in stochastic processes~\citep[see e.g.][Theorem 3.20]{harier-18}.

\begin{proposition}\label{prop2} Assume that $\mu+\gamma$ and $\beta$ are positive.  With probability one, $I_n=0$ for some $n\ge 1.$
\end{proposition}

To characterize metastability and extinction times, define $\s_+=\{(x_1,x_2)\in \s: x_2>0\}$ to be all the states where the disease persists. For all pairs of states $x,y\in\s_+$, let $Q_{xy}=\P[(S_{n+1},I_{n+1})=y|(S_n,I_n)=x]$ be the transition probabilities restricted to the transient states and $Q=(Q_{xy})_{x,y\in\s_+}$ be the associated matrix. Let $\pi=(\pi_x)_{x\in \s_+}$ be the quasi-stationary distribution with associated persistence probability $\lambda$ i.e. $\sum_{x\in\s_+} \pi_x=1$, $\pi_x>0$ for all $x\in\s$ and $\pi Q= \lambda \pi.$ Our main result for the stochastic SIR model concerns the behavior of the quasi-stationary distribution and the intrinsic mean time to extinction for large population size $N$.

\begin{theorem}\label{SIRSICH} Assume $\mu+\gamma>0,\beta>0$. For each $N\ge 1$, let $\pi^N$ be the quasi-stationary distribution and $\lambda^N$ the corresponding eigenvalue for \eqref{SSIR}. Let $\eig=\beta e^{-\mu}+e^{-\mu-\gamma}.$ Then
\begin{enumerate}
  \item[(i)]If $\eig<1$, then $\lambda^N\le \eig$ for all $N\ge 1$ and $\lim_{N\to\infty}\pi^N=\delta_{(1,0)}$ where $\delta_{(1,0)}$ is a Dirac measure at disease-free equilibrium $(1,0)$ and convergence is in the weak* topology.
  \item[(ii)] If $\eig>1$, then $\limsup_{N\to\infty}\frac{1}{N}\log (1-\lambda^N)<0$ and there exists a compact set $K\subset (0,1)^2$ such that $\pi^*(K)=1$ for every weak* limit point $\pi^*$ of $\pi^N$ as $N\to \infty,$ and where $\pi^*$ is invariant for the deterministic model \eqref{DSIR}. Moreover, if $(\mu,\beta,\gamma)\in C_P$, then $\lim_{N\to\infty}\mu^N=\delta_{(S^*,I^*)}$ where $\delta_{(S^*,I^*)}$ is the Dirac measure at the unique positive fixed point $(S^*,I^*)$ of equation~\eqref{DSIR}.
\end{enumerate}
\end{theorem}

The first assertion of Theorem~\ref{SIRSICH} implies that if $\eig<1$ and the population size is large, then any long-term transient mostly involves low frequencies of infected individuals and the mean time to extinction after these transients is less than $\frac{1}{1-\eig}.$ Furthermore, whenever permanence of the deterministic model corresponds to a globally stable equilibrium (see Theorem~\ref{SIR}), the QSDs concentrate on a Dirac measure at this equilibrium i.e. it supports the only invariant measure in $K$ for the deterministic dynamics. The second assertion of Theorem~\ref{SIRSICH} implies that if $\eig>1$ and the population size is large, then the long-term transients fluctuate away from the disease-free equilibrium and the mean time to extinction following these transients increases exponentially with population size i.e. there exist $c_1,c_2>0$ such that $\frac{1}{1-\lambda^N}\ge c_1 e^{c_2N}$ for all $N\ge 1.$ Figure~\ref{fig:SSIR} illustrates these conclusions.

\section{Discussion}

This paper formulates and provides a mathematically rigorous analysis of deterministic and stochastic, discrete-time SIS and SIR models. The stochastic models are based on probabilistic, individual-based update rules. The conditional expected change in the fraction of infected and susceptible individuals given the current values of these fractions determines the update rule for the deterministic model and, in this sense, the deterministic model is the mean field model for the stochastic models. 

Many earlier discrete-time epidemic models of SIS and SIR dynamics have been derived using numerical approximation schemes for differential equations~\citep{allen1994,ma2013global,satsuma2004extending,enatsu2010global,castillo2001discrete}. These models, including the one-dimensional ones, can exhibit oscillatory dynamics. In contrast, our model, which is based on individual-based update rules and uses an exponential escape function, is most similar to higher dimensional, discrete-time epidemiological models that have been used for applications to specific diseases~\citep{emmert2004population,emmert2006population,allen2008basic}. Unlike the models based on numerical approximation schemes, our analysis and numerical explorations suggest that our models always approach a global stable equilibrium. When $R_0\le 1$, we prove that all trajectories asymptotically approach the disease-free equilibrium for both the SIS and SIR models.  When $R_0>1$, we prove the disease persists for both models, approaches a globally stable, endemic equilibrium for the SIS model, and provide sufficient conditions for global stability of the endemic equilibrium for the SIR model. Extensive numerical simulations suggest that this endemic equilibrium of the SIR model is globally stable whenever $R_0>1$. Hence, we conjecture that $R_0>1$ always implies global stability of the endemic equilibrium for the SIR model. 

Unlike the deterministic model for which the disease persists indefinitely when $R_0>1$ and only goes asymptotically extinct over an infinite time horizon when $R_0\le 1$, the fraction of infected in the stochastic model becomes zero in finite time for all values of $R_0$. To understand this well-know discrepancy~\citep{bartlett1966introduction,Keeling2011,diekmann2012mathematical} for our model, we characterized the long-term transients using quasi-stationary distributions for finite-state, discrete-time Markov chains~\citep{darroch_seneta1965}. For these characterizations, we used the per-capita growth rate $\eig=\beta e^{-\mu}+ e^{-\mu-\gamma}$ of the infection at the disease free equilibrium.  When $\eig<1$ (equivalently, $R_0<1$), the mean time to extinction, when following the quasi-stationary distribution, is $\le 1/(1-\eig)$ for all population sizes and for both the SIS and SIR models. Indeed, we conjecture that as $N\to\infty$, this mean time to  extinction converges to $1/(1-\eig)$. While $R_0<1$ if and only if $\eig<1$, we have $\eig=(1-e^{-\mu-\gamma})R_0+e^{-\mu-\gamma}>R_0$ whenever $R_0<1$. Hence, even if $R_0$ is very small, the mean times to extinction can be arbitrarily  long. For example, given any $0<x<y<1$, we can make $R_0=x$ and $\eig=y$ by choosing $\gamma=0$, $e^{-\mu}=y/2$, and $\beta=x(1-e^{-\mu})$. 

When $R_0>1$ (equivalently $\eig>1$), we show that the mean extinction times increase exponentially with the population size $N$ and the quasi-stationary distributions concentrate on positive invariant sets for the deterministic model for large $N$. In particular, coupled with our analysis of the deterministic dynamics, our results imply that the quasi-stationary behavior for large $N$ always concentrates near the globally stable, endemic equilibrium of the SIS model. We provide sufficient conditions for the same conclusion for the SIR model, and conjecture that this always occurs for the SIR model. These conclusions are consistent with earlier studies of continuous-time Markov models where the analysis was done using diffusion approximations of the individual-based models~\citep{barbour1975,kryscio_lefevre1989,naasell1996,naasell1999,andersson2000stochastic,naasell2002,lindholm_britton2007,andersson2011stochastic,clancy_tjia2018}. In contrast, our results apply large deviation methods from \citep{FS} directly to the individual-based models. An open question for the stochastic model with $R_0>1$ concerns  the asymptotic rate at which the extinction times increase exponentially with $N$. Specifically, what is the value of $\alpha$ such that the mean time to extinction grows like $\exp(\alpha N)$ for large $N$? The diffusion approximations provide one approach to find potential candidates for $\alpha.$ 

When $R_0>1$,  we found that the mean time to extinction can be arbitrarily large even for a fixed population size. For example, this occurs when recovery and mortality rates approach zero in which case $R_0$ also increases without bound but $\eig$ remains bounded above by $\beta+1$. In contrast, increasing contact rates (which increase $\eig$ and $R_0$ without bound) leads to extinction times that are constrained by population size, recovery rates and mortality rates.

In addition to the open mathematical questions that we have already raised, future challenges include analyzing extensions of our models. These extensions could include additional compartments such as SEIR models, multi-age group epidemic models, and SIR type models with vaccination~\citep{anderson_may1991,Keeling2011,Kong2015}. More generally, when the discrete-time system is autonomous, the mathematical approaches used here should be  applicable to study quasi-stationary distributions and the intrinsic extinction times. However, when population sizes or transmission rate change stochastically over time \citep{Anderson1979,Pollicott2012}, new mathematical approaches are required for studying the impact of these environmentally driven random fluctuations on intrinsic extinction times.

\bigskip

\section{Proof of Theorem~\ref{SISDICH}}
Suppose that $\eig<1$. Then we shall prove that
\begin{equation}\label{ComP1}
  F(I):=e^{-\mu}(1-I)\big(1-e^{-\beta I}\big)+e^{-\mu-\gamma}I<  \big(\beta e^{-\mu}+ e^{-\mu-\gamma}\big)I =:L(I),\ I\in (0,1].
\end{equation}
It is easy to see that (\ref{ComP1}) is equivalent to
\begin{equation}\label{ComP2}
  (1-I)\big(1-e^{-\beta I}\big)< \beta I,\ I\in (0,1].
\end{equation}
Let $g(I):=(1-I)\big(1-e^{-\beta I}\big)-\beta I$. Then $g(0)=0$,
$$g'(I)=-(1+\beta)+e^{-\beta I}+\beta(1-I)e^{-\beta I}, \ g'(0)=0,$$
and
$$g''(I)=-\beta e^{-\beta I}\big(2+\beta(1-I)\big)<0,\ I\in [0,1].$$
Therefore,
$$g'(I)< g'(0)=0,\ I\in (0,1].$$
Furthermore,
$$g(I)< g(0)=0,\ I\in (0,1].$$
This shows that (\ref{ComP2}) holds.

Fix $I_0\in[0,1]$ and set $I_n:=F^n(I_0),\ J_n:=L^n(I_0)=\eig^nI_0,\ n=1, 2, \cdots.$ We claim that
\begin{equation}\label{ComP3}
  I_n\le J_n,\ n=1, 2, \cdots.
\end{equation}
For $n=1$, (\ref{ComP3}) follows immediately from (\ref{ComP1}). Assume that (\ref{ComP3}) holds for $n$. Then by (\ref{ComP1}) and the increasing of $L$, we have that
$$I_{n+1}=F(I_n)\le L(I_n)\le L(J_n)=J_{n+1}.$$
By mathematical induction, (\ref{ComP3}) is valid for all positive integers. Since $\eig<1$, $L^n(I_0)=\eig^n I_0 \rightarrow 0$ as $n\rightarrow 0$, i.e. $0$ is globally asymptotically stable.

 Suppose that $\eig=1$. Then it follows from (\ref{ComP1}) that $F(I)<I,\ I\in (0,I]$. Using Feigenbaum's method given in (ii), we can easily prove that $0$ is still globally asymptotically stable in this case.

(ii) Suppose that $\eig>1$. Then it is not difficult to prove that $F(I)$ has a unique positive fixed point, denoted by  $I^*$,  and  $F(I)$ has a unique positive critical point $I_c^*$. We will divide (ii) into two cases:
$${\rm (iia)}\ I^*\le I_c^*,\  \ \ {\rm (iib)}\ I^*>I_c^*.$$
We use Feigenbaum's method by depicting graphs to show that all nontrivial trajectories converge to the fixed point $I^*$.

(iia) Take $I_0\in (0, I^*)$. As depicted in Figure \ref{fig:1111}(a), the iterating sequence $I_n$ increasingly tends to $I^*$. If $I_0\in (I^*, I_c^*]$ with $I^*<I_c^*$, then Figure \ref{fig:1111}(b) shows that  the iterating sequence $I_n$ decreasingly tends to $I^*$. If $I_0 > I_c^*$, then the first iteration $I_1\in (0, I^*)$. After the first iteration, $I_n$ increasingly tends to $I^*$, see Figure \ref{fig:1111}(c).

(iib) It is easy to see that $F(I)$ is decreasing when $I>I_c^*$, that is, $F'(I)\le 0, \ I\in (I_c^*,1]$. By computation,
$$F'(I)\ge e^{-\mu}\big(-1+e^{-\beta I}+e^{-\gamma}\big)>-1.$$
Therefore, $-1<F'(I)\le 0, \ I\in (I_c^*,1]$.

First, we prove that $I_0<I_2<I^*$ if  $I_c^*<I_0<I^*$. Let $K$ be a line through $(I^*,F(I^*))$ whose gradient is $-1$. Because we have proved $F'(I)> -1$, $S=F(I)$ is between K and $S=I$ as shown in Figure  \ref{fig:2222}(a).

Plotting a line which is perpendicular to $I$-axis through $(I_0,0)$. This line intersects $S=I,S=F(I)$ and $K$  at $A,E$ and $B$, respectively. Through $B$ and $E$, we add lines parallel to $I$-axis, which intersect $S=I$ at $C$ and $G$, respectively. Because the ordinate of $B$ is strictly greater than the ordinate of $E$, the abscissa of $C$ is strictly greater than the abscissa of $G$. Sketch the lines perpendicular to $I$-axis through $C$ and $G$, which intersect $S=F(I)$ at $J$ and $H$, respectively. Because $S=F(I)$ is monotonically decreasing on $(I_c^*,1)$, the ordinate of $J$ are strictly less than the ordinate of $H$. Now we extend the segment $\overline{CJ}$ such that it intersects $K$ at $D$. The ordinate of $D$ is strictly less than the ordinate of $H$. Sketch the lines parallel to $I$-axis through $D$ and $H$. Then they intersect $S=I$ at $A$ and $L$, respectively. The ordinate of $L$ is strictly greater than the ordinate of $A$. Then the abscissa of $L$ is strictly greater than the abscissa of $A$, that is, $I_0\le I_2 \le I^*$ as shown Figure  \ref{fig:2222}(b).

Now, we want to show $\lim_{n \to \infty}F^{2n}(I_0)=I^*$ by contradiction. Suppose it is not true. Then $\lim_{n \to \infty}F^{2n}(I_0)=I^{**}< I^*$. Thus we set $I^{**}$ to be the initial point and repeat the above process, then we have $F^2(I^{**})> I^*$. This is contradictory to $\lim_{n \to \infty}F^{2n}(I_0)=I^{**}$, therefore, $\lim_{n \to \infty}F^{2n}(I_0)=I^*$. Similarly, we can prove $\lim_{n \to \infty}F^{2n+1}(I_0)=I^*$. This concludes that $I_n$ oscillates around $I^*$ and converges to $I^*$.

The proof of the case $I_0<I_c^*$ is given in the Figure \ref{fig:2222}(c).

\begin{figure}[!htb]
\centering
  \includegraphics[width=13cm]{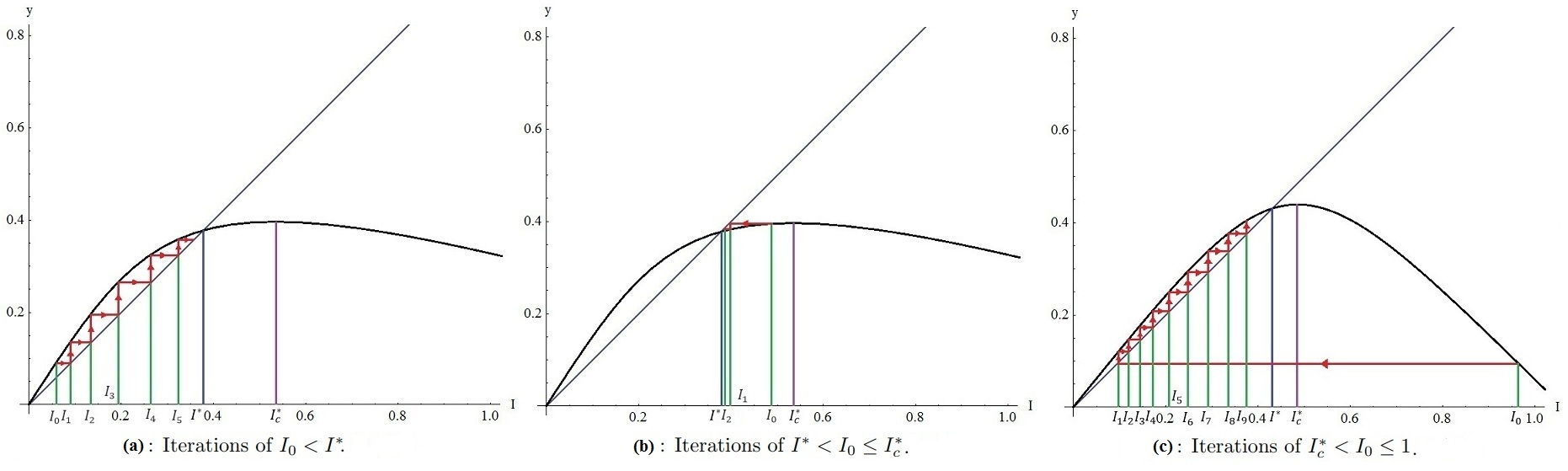}
  \caption{$I^* \leq I^*_c$.}\label{fig:1111}
\end{figure}

\begin{figure}[!htb]
\centering
  \includegraphics[width=13cm]{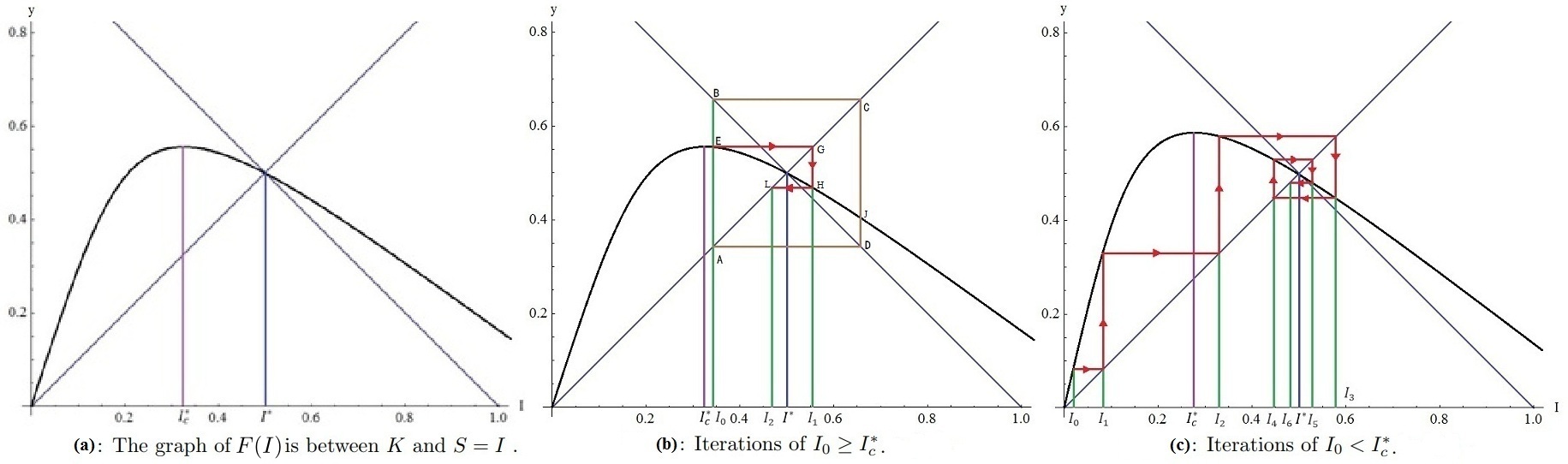}
  \caption{$I^* > I_c^*$ .}\label{fig:2222}
\end{figure}

\section{Proof of Theorem~\ref{SISSICH}}

We use results from \citet{FS} to prove Theorem~\ref{SISSICH}. We begin by verifying that Standing Hypothesis 1.1 of \citet{FS} holds. In their notation, ``$\varepsilon$'' corresponds to $\frac{1}{N}$ in our models i.e. small demographic noise corresponds to large population size $N$. For all $\delta>0$ and $N\ge 1$, define
\[
\beta_\delta(N)=\sup_{x\in[0,1]}\P\left[|I_{n+1}-F(x)|\ge \delta | I_n=x\right]
\]
where $F(x)=e^{-\gamma-\mu}x+e^{-\mu}(1-e^{-\beta x})(1-x)$ corresponds to the right hand side of the deterministic model in equation~\eqref{DSIS}. Standing Hypothesis 1.1 of \citet{FS} requires that $\lim_{N\to\infty} \beta_\delta(N)=0$ for all $\delta>0.$ The following proposition proves something stronger using large deviation estimates.

\begin{proposition}\label{prop:ld} There exists a function $\rho:(0,\infty)\to (0,\infty)$ such that
\[
\beta_\delta(N)\le \exp(-N \rho(\delta))
\]
for all $N\ge 1$ and $\delta>0.$
\end{proposition}

\begin{proof} While we use standard large deviation estimates, we go through the details to ensure that the estimates can be taken uniformly in $x\in [0,1].$
Define $a=e^{-\mu-\gamma}$ and $b(x)=e^{-\mu}(1-e^{-\beta x})$. By the exponential Markov inequality, we have for all $t$
\begin{equation}\label{ld0}
\begin{aligned}
\P[N(I_{n+1}-F(x))&\le N\delta|I_n=x]\le e^{-tN\delta}\E[e^{t(N(I_{n+1}-F(x)))}|I_n=x]\\
&=e^{-tF(x)N-t\delta N}\left(1-a+ae^t\right)^{Nx}\left(1-b(x)+b(x)e^t\right)^{N(1-x)}
\end{aligned}
\end{equation}
Define
\[
\psi(x,t)=-tF(x)+x\log(1-a+ae^t)+(1-x)\log(1-b(x)+b(x)e^t).
\]
Taking $\log$ of equation~\eqref{ld0} and dividing by $N$ yields
\begin{equation}\label{ld1}
\frac{1}{N}\log \P[N(I_{n+1}-F(x))\ge N\delta|I_n=x]\le -\delta t+ \psi(t,x)
\end{equation}{}
for all $t\neq 0$, $x\in [0,1]$, and $\delta>0$. Similarly, one can show that
\begin{equation}\label{ld2}
\frac{1}{N}\log \P[N(F(x)-I_{n+1})\ge N\delta|I_n=x]\le -\delta t+ \psi(-t,x)
\end{equation}
for all $t\neq 0$, $x\in [0,1]$, and $\delta>0$. We have
\[
\frac{\partial \psi}{\partial t}=-F(x)+\frac{x a e^t}{1-a+ae^t}+\frac{(1-x)b(x)e^t}{1-b(x)+b(x)e^t}
\]
and
\[
\frac{\partial^2\psi}{\partial t^2}=ax\frac{ e^t(1-a)}{(1-a+ae^t)^2}+(1-x)b(x)\frac{e^t(1-b(x))}{1-b(x)+b(x)e^t}>0.
\]
As $\psi(0,x)=0=\frac{\partial \psi}{\partial t}(0,x)=0$, and $\psi$ is strictly convex in $t$, for all $\delta>0$, there exists $t^*(\delta)>0$ such that $\delta t^*(\delta)>\psi(-t^*(\delta),x)$ and $\delta t^*(\delta)>\psi(t^*(\delta),x)$ for all $x\in [0,1].$ Define
\[
\rho(\delta)=\delta t^*(\delta)-\max_{x\in[0,1]}\psi(t^*(\delta),x).
\]
Equations \eqref{ld1}--\eqref{ld2} imply that
\[
\P[|I_{n+1}-F(x)|\ge \delta|I_n=x]\le \exp(-N\rho(\delta))
\]
for all $x\in [0,1]$ and $\delta>0.$
\end{proof}

To prove the first result of Theorem~\ref{SISSICH}, assume that $\eig\le 1$. Theorem~\ref{SISDICH} implies that $0$ is globally stable for the deterministic model $I\mapsto F(I)$. Theorem 3.12 of \citet{FS}, which only requires Standing Hypothesis 1.1, implies that $\lim_{N\to\infty}\pi^N=\delta_0.$ Define $R(x)=F(x)/x$ for $x\in (0,1]$. Equation~\eqref{ComP1} in the proof of Theorem~\ref{SISDICH} implies that $R(x)\le \eig$ for $x\in (0,1].$ For $N\ge 1$, quasi-stationarity of $\pi^N$ implies
\[
\begin{aligned}
\lambda^N\sum_{i=1}^N\frac{i}{N} \pi_i^N=&\sum_{i=1}^N \E\left[I_{n+1}\Big|I_n=\frac{i}{N}\right] \pi_i^N\\
=&\sum_{i=1}^N F\left(\frac{i}{N}\right) \pi_i^N\\
=&\sum_{i=1}^N \frac{i}{N}R\left(\frac{i}{N}\right) \pi_i^N\\
\le&\eig \sum_{i=1}^N  \frac{i}{N} \pi_i^N.
\end{aligned}
\]
Since $\sum_{i=1}^N \frac{i}{N}\pi_i^N>0$, $\lambda^N\le \eig$ for all $N\ge 1$ as claimed.

To prove the second result of Theorem~\ref{SISSICH}, assume $\eig>1$ in which case Theorem~\ref{SISDICH} implies that there exists a unique positive globally stable equilibrium $I^*$ for the map $I\mapsto F(I)$. Assertion (b) of  Lemma 3.9 of \citet{FS}, implies that there exists $\delta>0$ such that $1-\lambda^N\le \beta_\delta(N)$ for all $N\ge 1$. Proposition~\ref{prop:ld} implies that
\[
\limsup_{N\to\infty}\frac{1}{N}\log(1-\lambda^N)\le -\rho(\delta).
\]

To complete the proof of the second assertion, we need to verify the assumption in Assertion (b') of Lemma 3.9 of \citet{FS}. Choose $\eta>0$ sufficiently small so that
\[
\min_{x\in [0,\eta]}(1-e^{-\mu-\gamma})^x\ge \exp(-\rho(\delta)/3) \mbox{ and }\min_{x\in[0,\eta]}
(1-e^{-\mu}(1-e^{-\beta x}))^{1-x}\ge \exp(-\rho(\delta)/3).
\]
Then
\[
\min_{x\in[0,\eta]}\P[I_{n+1}=0|I_n=x]\ge \exp(-2N\rho(\delta)/3)
\]
and
\[
\lim_{N\to\infty}\frac{\beta_\delta(N)}{\min_{x\in[0,\eta]}\P[I_{n+1}=0|I_n=x]}\le \lim_{N\to\infty}\frac{ \exp(-N\rho(\delta))}{ \exp(-2N\rho(\delta)/3)}=0
\]
which verifies the assumption of (b') of Lemma 3.9 of \citet{FS} and implies that
\[
\lim_{N\to\infty}\sum_{i/N \le \eta}\pi_i^N=0
\]
i.e. for any weak* limit point $\pi^*$ of $\pi^N$, $\pi^*([0,\eta])=0$. As $\lambda^N\to 1$, Proposition 3.11 of \citet{FS} implies that any weak* limit point of $\pi^N$ is invariant for the dynamics $x\mapsto F(x)$. As these weak* limit points are supported on $[\eta,1]$ and the only invariant measure for $x\mapsto F(x)$ on this interval is the Dirac measure $\delta_{I^*}$ and the unique positive fixed point, it follows that $\pi^N$ converges in the weak* topology to $\delta_{I^*}$ as claimed.

\section{Proof of Theorem~\ref{SIR}}
Let $E_f=(1,0)$. Denote by
\begin{displaymath}
  F(S,I):=\left (
         \begin{array}{cc}
             1-e^{-\mu}+S e^{-\mu-\beta I} \\
            e^{-\mu}S\big(1-e^{-\beta I}\big)+e^{-\mu-\gamma}I
        \end{array}
    \right).
\end{displaymath}
Then
\begin{displaymath}
  DF(E_f):=\left (
    \begin{array}{ccc}
        e^{-\mu} & -\beta e^{-\mu}\\
        0 & \beta e^{-\mu}+ e^{-\mu-\gamma}
    \end{array}
    \right ).
\end{displaymath}
It follows that the per-capita growth rate of the disease is still given by (\ref{BRN}) and  the basic reproduction number of (\ref{DSIR}) is still the expression in (\ref{BRN2}). Define the competitive cone
$$K:=\{(u,v):u \le 0,\ v\ge 0\}.$$
Then it is easy to check that $ DF(E_f)$ keeps $K$ invariant (see \citet{WJ}). Define
\begin{displaymath}
            L(S,I):=\left (
         \begin{array}{cc}
            1 \\
            0
        \end{array}
    \right)+\left (
    \begin{array}{cc}
     e^{-\mu} & -\beta e^{-\mu}\\
        0 & \beta e^{-\mu}+ e^{-\mu-\gamma}
    \end{array}
    \right )\left (
         \begin{array}{cc}
            S-1 \\
            I
        \end{array}
    \right).
\end{displaymath}
We shall verify that
\begin{equation}\label{ComP4}
  F(S,I)\le_K L(S,I),\ (S,I)\in X.
\end{equation}
(\ref{ComP4}) is equivalent to
\begin{equation}\label{ComP5}
    \left\{\begin{array}{l}
                 1-e^{-\mu}+S e^{-\mu-\beta I}\ge 1 +e^{-\mu}(S-1)-\beta Ie^{-\mu}, \\
                 \noalign{\medskip}
                  e^{-\mu}S\big(1-e^{-\beta I}\big)+e^{-\mu-\gamma}I\le \big(\beta e^{-\mu}+ e^{-\mu-\gamma}\big)I \\
 \end{array}\right.
\end{equation}
on $X$. By simple computation, (\ref{ComP5}) is equivalent to
\begin{equation}\label{ComP6}
    S(1-e^{-\beta I})\le \beta I,\ (S,I)\in X.
\end{equation}
Since $(S,I)\in X$, $S\le (1-I)$. Thus (\ref{ComP6}) follows immediately from (\ref{ComP2}), that is, the competitive ordering relation (\ref{ComP4}) holds.

Let $P_0:=(S_0,I_0)\in X,\ P_n:=F^n(P_0)=(S_n(P_0),I_n(P_0)),\ Q_n:=L^n(P_0)$. We shall show that
\begin{equation}\label{ComP7}
(1,0)^{\tau}\le_K P_n\le_K Q_n,\ n=1, 2,\ \cdots.
\end{equation}
The left inequality is obvious by the definition of competitive order and $X$. So we will prove the right one. (\ref{ComP4}) deduces that (\ref{ComP7}) holds for $n=1$. Suppose that (\ref{ComP7}) holds for $n$. Then using (\ref{ComP4}) and the order-preserving for $DF(E_f)$ , we obtain
$$P_{n+1}=F(P_n)\le_K L(P_n)\le_K L(Q_n) =Q_{n+1}.$$
By mathematical induction, (\ref{ComP7}) holds.

Let $\eig<1$. Then $Q_n\rightarrow (1,0)$ as $n\rightarrow \infty$. Therefore, (i) is proved by (\ref{ComP7}).

Suppose $\eig>1$. Then we shall prove that the system (\ref{DSIR}) is uniformly persistent with respect to $(X_0,\partial{X_0})$, that is, there exists $\eta>0$ such that
\begin{equation}\label{up}
\liminf_{n\rightarrow \infty}I_n(P_0)\ge \eta,\ {\rm for\ all}\ P_0=(S_0,I_0)\in X_0.
\end{equation}
 It is easy to see that $E_f$ is the
maximal compact invariant set in $\partial{X_0}$ which is positively invariant with respect to $F$ and lies on the stable manifold of $E_f$. Recalling the Hofbauer and So uniform persistence theorem of \citet{HS}, the system (\ref{DSIR}) is uniformly persistent if and only if

(a) $E_f$ is isolated in $X$, and

(b) $W^s(E_f)\subset \partial{X_0}$.\\
Since the disease free fixed point $E_f$ is hyperbolic and a saddle, (a) and (b) follows immediately from
the Hartman and Grobman theorem (see \citet{GH}). This verifies the uniform persistence and hence the system (\ref{DSIR}) admits an attractor in $X_0$.

It follows from (\ref{up}) that (\ref{DSIR}) contains a compact attractor $A_0\subset \{(S,I)\in X: I\ge \eta\}$. Besides, by the first equality of (\ref{DSIR}), we have $ S_n(P_0)\ge 1-e^{-\mu}$ for $n\ge 1$ and $P_0\in X$. This implies that $A_0\subset \{(S,I)\in X: S\ge 1-e^{-\mu}\}$. As  a result,
$$A_0\subset \{(S,I)\in X: S\ge 1-e^{-\mu},\ I\ge \eta,\ S+I \le 1\}.$$
This proves (ii).

It remains to prove (iii). First, we consider the system (\ref{DSIR}) with $\lambda_0=(\mu_0,\beta_0,0)$ and $\eig(\mu_0,\beta_0,0)>1$. We shall prove that it admits a globally stable fixed point $(1-I^*,I^*)$ in $X$, where $F(I^*)=I^*$ with $0<I^*<1$.

Let $\Delta_n:=S_n+I_n$. Then from (\ref{DSIR}) it follows that
\begin{equation}\label{sum}
\Delta_{n+1}=1-e^{-\mu_0}+e^{-\mu_0}\Delta_n.
\end{equation}
It is easy to see that (\ref{sum}) has the positive fixed point $1$ and all positive orbits tend to $1$ as $n\rightarrow +\infty$. Therefore, the system (\ref{DSIR}) is reduced to the system (\ref{DSIS}) with $\mu=\mu_0, \beta=\beta_0, \gamma=0$. Applying Theorem \ref{SISDICH}(ii), we get that the system (\ref{DSIS}) has a globally stable fixed point $I^*$ in $(0,1)$. Thus the system (\ref{DSIR}) admits a globally stable fixed point $(1-I^*,I^*)$ in $X$, where $F(I^*)=I^*$ with $0<I^*<1$. Recalling the proof of Theorem \ref{SISDICH}(ii)(see Figure \ref{fig:1111} and Figure \ref{fig:2222}(a)), we have
\begin{equation}\label{gradient}
|F'(I^*)|=\big|e^{-\mu_0-\beta_0I^*} \big(1+\beta_0(1-I^*)\big)\big|<1.
\end{equation}

Next, we will prove the spectral radius of the Jacobian matrix for $F(S,I)$ at the positive fixed point $E^*(S^*,I^*):=(1-I^*,I^*)$ is strictly less than $1$.

An easy calculation yields that
\begin{displaymath}
  DF(E^*):=\left (
    \begin{array}{ccc}
        e^{-\mu_0-\beta_0I^*} & -\beta_0S^* e^{-\mu_0-\beta_0I^*}\\
         e^{-\mu_0}(1-e^{-\beta_0I^*}) &  \beta_0S^* e^{-\mu_0-\beta_0I^*}+e^{-\mu_0}
    \end{array}
    \right ),
\end{displaymath}

\begin{displaymath}
    \left (
    \begin{array}{ccc}
        e^{-\mu_0-\beta_0I^*} & -\beta_0S^* e^{-\mu_0-\beta_0I^*}\\
         e^{-\mu_0}(1-e^{-\beta_0I^*}) &  \beta_0S^* e^{-\mu_0-\beta_0I^*}+e^{-\mu_0}
    \end{array}
    \right )\left (
         \begin{array}{cc}
            1 \\
            -1
        \end{array}
    \right)=F'(I^*)\left (
         \begin{array}{cc}
            1 \\
            -1
        \end{array}
        \right)
\end{displaymath}
and ${\rm det}DF(E^*)=e^{-\mu_0}F'(I^*)$. This proves that $F'(I^*)$ and $e^{-\mu_0}$ are two eigenvalues of $DF(E^*)$, that is, the spectral radius of $DF(E^*)$ is strictly less than $1$.

    In what follows, we shall use Theorem 2.1 of \cite{SW} to  prove that $C_P$ is open in the  parameter space $P$.

For this purpose, denote by $\|\cdot\|$ the Euclidean norm of $\mathbb{R}^3$ and $B_C(\lambda_0,s)$ the open ball in $P$ of radius $s$ about the point $\lambda_0$. We fix a $\lambda_0\in C_P$ and an $s_0\in (0,\mu_0)$ such that $\eig(\lambda)>1$ for any $\lambda\in B_C(\lambda_0,s_0)$. In order to consider the perturbed systems for parameters, we set $F_{\lambda_0}(S,I)$ the given mapping and $F_{\lambda}(S,I)=(S_{\lambda}(S,I),I_{\lambda}(S,I))$ the mappings for $\lambda\in B_C(\lambda_0,s_0)$. Define
\[R_{\lambda}(S,I)=\begin{cases}
\frac{I_{\lambda}(S,I)}{I},\quad {\rm if}\ I>0,\\
\beta e^{-\mu}S+e^{-\mu-\gamma}, \quad {\rm if}\ I=0.
\end{cases} \]
Then $R_{\lambda}(S,I)$ is continuous on $X$. By induction, it not difficult to prove that
\begin{equation}\label{niteration}
\begin{aligned}
F^n_{\lambda}(S,0)=&\big(1-e^{-n\mu}+e^{-n\mu}S,0\big)\\
R_{\lambda}(F^n_{\lambda}(S,0))=& \beta e^{-\mu}\big(1-e^{-n\mu}+e^{-n\mu}S\big)+ e^{-\mu-\gamma}
\end{aligned}
\end{equation}
for $n=1,2,\cdots$. We claim that there exist an $s_1\in (0,s_0]$, an integer $N>0$, $\delta>0$  and $\rho>1$, all only depending on $\lambda_0$, such that
\begin{equation}\label{niteration1}
I^N_{\lambda}(S,I)\ge \rho I\ {\rm for\ all}\ \lambda\in B_C(\lambda_0,s_1)\ {\rm and}\ I\in [0,\delta],\\
\end{equation}
where $F^n_{\lambda}(S,I)=(S^n_{\lambda}(S,I),I^n_{\lambda}(S,I))$.

From (\ref{niteration}) it follows that
$$R_{\lambda}(F^n_{\lambda}(S,0))\ge \eig(\lambda)-\beta e^{-(n+1)\mu}.$$
By the continuity of $\eig(\lambda)$, there exists an $s_1\in (0,s_0]$ such that $\eig(\lambda)> \frac{\eig(\lambda_0)+1}{2}$ for all $\lambda\in B_C(\lambda_0,s_1)$, and hence
$$R_{\lambda}(F^n_{\lambda}(S,0))\ge \frac{\eig(\lambda_0)+1}{2}-(\beta_0+s_0) e^{-(n+1)(\mu_0-s_0)}\ {\rm for\ all}\ \lambda\in B_C(\lambda_0,s_1).$$
This implies that there is an integer $N>0$, only depending on $\lambda_0$, such that
\begin{equation}\label {Niteration}
R_{\lambda}(F^N_{\lambda}(S,0)) > \frac{\eig(\lambda_0)+3}{4}\ {\rm for\ all}\ \lambda\in B_C(\lambda_0,s_1).
\end{equation}
Using (\ref{Niteration}) and  the uniform continuity of $R_{\lambda}(F^N_{\lambda}(S,I))$ on $\overline{B_C(\lambda_0,s_1)}\times X$, we obtain that there is a $\delta>0$, only depending on $\lambda_0$, such that
\[
R_{\lambda}(F^N_{\lambda}(S,I)) > \frac{\eig(\lambda_0)+7}{8}:=\rho\ {\rm for\ all}\ \lambda\in B_C(\lambda_0,s_1)\ {\rm and}\ 0\le I \le \delta,
\]
which implies that (\ref{niteration1}) holds, thus the claim is proved.

By (\ref{niteration1}), we get that for each $\lambda\in B_C(\lambda_0,s_1)$,
\[I^{mN}_{\lambda}(S,I)\ge \rho^mI\ {\rm if}\ F^{kN}_{\lambda}(S,I)\in [0,1]\times (0,\delta]\ {\rm for}\ k=0,1,\cdots,m-1.  \]
This shows that there exists at least a positive integer $m$ with the property
\[F^{mN}_{\lambda}(S,I)\in [0,1]\times (\delta,1]\ {\rm if}\ (S,I)\in [0,1]\times (0,\delta].  \]
Let $U=X_0,\ \Lambda= B_C(\lambda_0,s_1),\ B_{\lambda}=[0,1]\times [\delta,1]$. Then all assumptions of Theorem 2.1 of \cite{SW} have been checked. It follows that $B_C(\lambda_0,s_1)\subset C_P$. The proof is complete.
\section{Proof of Theorem~\ref{SIRSICH}}

As in the proof of Theorem~\ref{SISSICH}, we use results from \citet{FS} to prove Theorem~\ref{SIRSICH}. We begin by verifying that Standing Hypothesis 1.1 of \citet{FS} holds.  For all $\delta>0$ and $N\ge 1$, define
\[
\beta_\delta(N)=\sup_{x,y\ge 0, x+y\le 1}\P\left[\|(S_{n+1},I_{n+1}-F(x,y)\|_\infty \ge \delta | (S_n,I_n)=(x,y)\right]
\]
where $F(x,y)=(1-e^{-\mu}+e^{-\mu-\beta y}x,
xe^{-\mu}(1-e^{-\beta y})+ye^{-\mu-\gamma})$ corresponds to the right hand side of the deterministic model in equation~\eqref{DSIR} and $\|(x,y)\|_\infty=\max\{|x|,|y|\}$ corresponds to the sup norm. Standing Hypothesis 1.1 of \citep{FS} requires that $\lim_{N\to\infty} \beta_\delta(N)=0$ for all $\delta>0.$ Like Proposition~\ref{prop:ld} for the stochastic SIS model, the following proposition proves something stronger using large deviation estimates.

\begin{proposition}\label{prop2:ld} There exists a function $\rho:(0,\infty)\to (0,\infty)$ such that
\[
\beta_\delta(N)\le \exp(-N \rho(\delta))
\]
for all $N\ge 1$ and $\delta>0.$
\end{proposition}

\begin{proof}
We begin by observing that $NS_n-W_{n+1}$ and $X_{n+1}+Z_{n+1}$ in \eqref{SSIR} conditioned on $(S_n,I_n)=(x,y)$ are independent binomials where $NS_n-W_{n+1}$ has $Nx$ trials with probability of success $a_1(y)=1-e^{-\mu}(1-e^{-\beta y})$ and $X_{n+1}+Z_{n+1}$ has $N(1-x)$ trials with probability of success $b_1=1-e^{-\mu}.$  Using the exponential Markov inequality as in the proof of Proposition~\ref{prop:ld}, we get
\begin{equation}\label{ld21}
\begin{aligned}
\frac{1}{N}\log \P[N(S_{n+1}-F_1(x,y))\ge& N\delta|(S_n,I_n)=(x,y)]\le -\delta t+ \psi_1(t,x,y)\\
\frac{1}{N}\log \P[N(F_1(x,y)-S_{n+1})\ge & N\delta|(S_n,I_n)=(x,y)]\le -\delta t+ \psi_1(-t,x,y)
\end{aligned}
\end{equation}
for all $t,\delta$ and where
\[
\psi_1(t,x,y)=-tF_1(x,y)+x\log(1-a_1(y)+a_1(y)e^t)+(1-x)\log(1-b_1+b_1e^t).
\]
As $\psi_1(0,x,y)=0=\frac{\partial \psi}{\partial t}(0,x,y)=0$, and $\psi_1$ is strictly convex in $t$, for all $\delta>0$, there exists $t^*(\delta)>0$ such that $\delta t^*(\delta)>\psi_1(-t^*(\delta),x,y)$ and $\delta t^*(\delta)>\psi_1(t^*(\delta),x,y)$ for all $x,y\in [0,1]$ such that $x+y\le 1$. Define
\[
\rho_1(\delta)=\delta t^*(\delta)-\max_{x,y\in[0,1], x+y\le 1}\psi_1(t^*(\delta),x,y)>0.
\]
Equation \eqref{ld21} implies that
\[
\P[|S_{n+1}-F_1(x,y)|\ge \delta|(S_n,I_n)=(x,y)]\le \exp(-N\rho_1(\delta))
\]
for all $x,y\in [0,1]$ with $x+y\le 1$ and $\delta>0.$

$W_{n+1}$ and $Y_{n+1}$ conditioned on $(S_n,I_n)=(x,y)$ in equation \eqref{SSIR} are also independent binomial random variables where $W_{n+1}$ has $Nx$ trials with probability of success $e^{-\mu}(1-e^{-\beta y})$ and $Y_{n+1}$ has $Ny$ trials with probability of success $e^{-\gamma}$. Therefore using the exponential Markov inequality, one can show there exists a function $\rho_2:(0,\infty)\to(0,\infty)$ such that
\[
\P[|I_{n+1}-F_2(x,y)|\ge \delta|(S_n,I_n)=(x,y)]\le \exp(-N\rho_2(\delta))
\]
for all $x,y\in [0,1]$ with $x+y\le 1$ and $\delta>0.$ Setting $\rho(\delta)=\min\{\rho_1(\delta),\rho_2(\delta)\}$
completes the proof of the proposition.
\end{proof}

To prove the first result of Theorem~\ref{SIRSICH}, assume that $\eig\le 1$. Theorem~\ref{SIR} implies that $(1,0)$ is globally stable for the deterministic model $(S,I)\mapsto F(S,I)$. Theorem 3.12 of \citet{FS}, which only requires Standing Hypothesis 1.1, implies that $\lim_{N\to\infty}\pi^N=\delta_{(1,0)}$ in the weak* topology. Define $R(x,y)=F_2(x,y)/y$ for $y\in (0,1],x\in[0,1]$, and $x+y\le 1$. Equation~\eqref{ComP5} in the proof of Theorem~\ref{SIR} implies that $R(x,y)\le \eig$. For $N\ge 1$, quasi-stationarity of $\pi^N$ implies
\[
\begin{aligned}
\lambda^N\sum_{x,y\in\s_+}y \pi_{x,y}^N=&\sum_{x,y\in\s_+}\E\left[I_{n+1}\Big|(S_n,I_n)=(x,y)\right] \pi_{x,y}^N\\
=&\sum_{x,y\in\s_+} F_2\left(x,y\right) \pi_{x,y}^N\\
=&\sum_{x,y\in\s_+} y R\left(x,y\right) \pi_{x,y}^N\\
\le&\eig \sum_{x,y\in\s_+}  y \pi_{x,y}^N
\end{aligned}
\]
Since $\sum_{x,y\in\s_+}y\pi_{x,y}^N>0$, $\lambda^N\le \eig$ for all $N\ge 1$ as claimed.

To prove the second result of Theorem~\ref{SIRSICH}, assume $\eig>1$ in which case Theorem~\ref{SIR} implies that there exists a global, compact attractor $K\subset (0,1)\times (0,1)$ for $(S,I)\mapsto F(S,I)$. Assertion (b) of  Lemma 3.9 of \citet{FS}, implies that there exists $\delta>0$ such that $1-\lambda^N\le \beta_\delta(N)$ for all $N\ge 1$. Proposition~\ref{prop2:ld} implies that
\[
\limsup_{N\to\infty}\frac{1}{N}\log(1-\lambda^N)\le -\rho(\delta).
\]
Assumption in Assertion (b') of Lemma 3.9 of \citet{FS} holds for an argument similar to the proof of Theorem~\ref{SISSICH} and consequently any weak* limit point $\pi^*$ of $\pi^N$ satisfies $\pi^*(K)=1$. As $\lambda^N\to 1$, Proposition 3.11 of \citet{FS} implies that any weak* limit point of $\pi^N$ is invariant for the dynamics $(x,y)\mapsto F(x,y)$.

\bibliographystyle{EcologyLetters}

\bibliography{references}
\end{document}